\documentclass[a4paper,10pt,onecolumn]{llncs}
\pdfoutput=1
\def\PubType{Conference}

\usepackage[utf8]{inputenc}
\usepackage[english]{babel}  
\usepackage[T1]{fontenc} 
\usepackage{IEEEtrantools}
\usepackage{epsfig}
\usepackage{amsmath, amssymb, amsbsy}
\usepackage{xspace}          
\usepackage[noend]{algpseudocode}
\usepackage{algorithm}
\usepackage{algorithmicx}
\usepackage{color}
\usepackage{cite}            
\usepackage{booktabs} 
\usepackage{verbatim} 
\usepackage{hyperref}

\definecolor{navy}{rgb}{0,0,.5}

\usepackage[final,tracking=true,kerning=true,spacing=true,factor=1100,stretch=10,shrink=10]{microtype}

\algrenewcommand\alglinenumber[1]{{\scriptsize#1}}
\algrenewcommand\algorithmicrequire{\textbf{Input:}}
\algrenewcommand\algorithmicensure{\textbf{Output:}}
\newcommand{\Ifline}[2]{\State \textbf{if }#1{ \textbf{then} }#2}

\usepackage{varioref}        
\usepackage{fancyref}

\makeatletter
\def\mkfancyprefix#1#2{%
\expandafter\def\csname fancyref#1labelprefix\endcsname{#1}%
\begingroup\def\x{\endgroup\frefformat{plain}}%
    \expandafter\x\csname fancyref#1labelprefix\endcsname
    {\MakeLowercase{#2}\fancyrefdefaultspacing##1}%
\begingroup\def\x{\endgroup\Frefformat{plain}}%
    \expandafter\x\csname fancyref#1labelprefix\endcsname
    {#2\fancyrefdefaultspacing##1}%
\begingroup\def\x{\endgroup\frefformat{vario}}%
    \expandafter\x\csname fancyref#1labelprefix\endcsname
    {\MakeLowercase{#2}\fancyrefdefaultspacing##1##3}%
\begingroup\def\x{\endgroup\Frefformat{vario}}%
    \expandafter\x\csname fancyref#1labelprefix\endcsname
    {#2\fancyrefdefaultspacing##1##3}%
}
\makeatother
\fancyrefchangeprefix{\fancyrefeqlabelprefix}{eqn}
\mkfancyprefix{ssec}{Section}
\mkfancyprefix{tbl}{Table}
\mkfancyprefix{thm}{Theorem}
\mkfancyprefix{lem}{Lemma}
\mkfancyprefix{cor}{Corollary}
\mkfancyprefix{prop}{Proposition}
\mkfancyprefix{prob}{Problem}
\mkfancyprefix{alg}{Algorithm}
\mkfancyprefix{inv}{Invariant}
\mkfancyprefix{ex}{Example}
\mkfancyprefix{line}{Line}
\mkfancyprefix{def}{Definition}
\mkfancyprefix{itm}{Item}
\newcommand{\cref}[1]{\Fref{#1}}

\def\ve#1{{\mathchoice{\mbox{\boldmath$\displaystyle #1$}}%
              {\mbox{\boldmath$\textstyle #1$}}%
              {\mbox{\boldmath$\scriptstyle #1$}}%
              {\mbox{\boldmath$\scriptscriptstyle #1$}}}}

\def\confType{Conference} 
\def\jourType{Journal}

\definecolor{orange}{rgb}{1,0.5,0}
\ifx\PubType\confType
  \newcommand{\todo}[1]{}

\else\ifx\PubType\jourType
  \newcommand{\todo}[1]{}

\else
  \newcommand{\todo}[1]{{\color{red}[#1]}}

\fi\fi

\newcommand{\param}{\mu}
\newcommand{\K}{K}

\newcommand{\F}{\mathbb{F}}
\newcommand{\R}{{\cal{R}}}
\newcommand{\Q}{{\cal{Q}}}
\newcommand{\N}{\mathbb{N}}
\newcommand{\Z}{\mathbb{Z}}
\newcommand{\Qab}{\Q^{\mathrm{ab}}}
\newcommand{\M}{{\cal{M}}}
\newcommand{\V}{{\cal{V}}}
\newcommand{\MQ}{\Q^{n\times n}} 
\newcommand{\GLQ}{\mathrm{GL}_n(\Q)}
\newcommand{\OD}[1]{{\Delta(#1)}} 
\newcommand{\m}{\ensuremath{\ve{m}}}
\renewcommand{\v}{\ensuremath{\ve{v}}}
\renewcommand{\u}{\ensuremath{\ve{u}}}
\newcommand{\LPc}{\ensuremath{\mathrm{LP}}}
\newcommand{\LP}[1]{\LPc(#1)}
\newcommand{\LT}[1]{\ensuremath{\mathrm{LT}(#1)}}
\newcommand{\LC}[1]{\ensuremath{\mathrm{LC}(#1)}}

\newcommand{\word}[1]{\textnormal{#1}}
\newcommand\modop{\ \word{mod}\ }         
\newcommand{\code}[1]{\textup{\textsf{#1}}}

\newcommand\vecVal{\psi}
\newcommand\Dieu{Dieudonn\'e\xspace}

\DeclareMathOperator{\diag}{diag}
\DeclareMathOperator{\maxdeg}{maxdeg}

\AtBeginDocument{}

\begin{document}

\title{Solving Shift Register Problems over Skew Polynomial Rings using Module
Minimisation}

\author{W. Li\inst{1}, J.S.R. Nielsen\inst{2}, S. Puchinger\inst{1}, V. Sidorenko\inst{1}\inst{3}}

\institute{
Institute of Communications Engineering, Ulm University, Germany\\
\email{\{wenhui.li | sven.puchinger | vladimir.sidorenko\}@uni-ulm.de}
\and
GRACE Project, INRIA Saclay \& LIX, \'Ecole Polytechnique, France\\
\email{jsrn@jsrn.dk}
\and
Institute for Communications Engineering, TU München, Germany}

\maketitle

\begin{abstract}
For many algebraic codes the main part of decoding can be reduced to a shift register synthesis problem.
In this paper we present an approach for solving generalised shift register problems over skew polynomial rings which occur in error and erasure decoding of $\ell$-Interleaved Gabidulin codes.
The algorithm is based on module minimisation and has time complexity $O(\ell \param^2)$ where $\param$ measures the size of the input problem.
\end{abstract}

\section{Introduction}

\todo{/jsrn: Some thoughts for improvements for journal version:
  \begin{itemize}
    \item Prove in M-S applied to $M$ that no poly on first column has degree exceeding degree of minimal solution. This implies perhaps fewer simple transformations performed. For D-D it implies smaller degree polys to add and convolute, improving complexity (replacing $\param$ by degree of minimal solution or something.)
    \item In D-D for $g_i$ not satisfying the sparse-etc. requirement, I think we can improve complexity down to $O(\ell \param^2 \log \param)$ as follows (this is faster than M-S when $\log \param \in o(\ell)$.): the operation $\lambda_i s_h \modop g_h$ can be represented as a (symbolic) linear transformation $\ve{\lambda_i}D_\theta\amalg_{s_h}\Xi_{g_h}$, where:
    \begin{itemize}
      \item $D_\theta$ is a diagonal matrix of $\theta$ powers, i.e. $D_\theta = \diag(1, \theta, \ldots, \theta^m)$.
      \item $\amalg_{s_h}$ is the Toplitz matrix where each row is the coefficients of $s_h$.
      \item $\Xi_{g_h}$ is a matrix representing the modulo reduction, i.e. for instance, the $n$th column for $n \geq \deg g_h$ has the coefficients of $x^n \mod g_h$ as its entries and $0$ on all entries $\geq \deg g_h$.
    \end{itemize}
    The point is that $\Xi_{g_h}$ can be precomputed (increasing, alas, the memory requirement to $O(\ell \param^2)$).
    To obtain only one coefficient of $\lambda_h s_h \modop g_h$ we need only one column of $\amalg_{s_h} \Xi_{g_h}$, i.e. we need $\amalg_{s_h} \xi_{g_h,t}$ for some column $\xi_{g_h, t}$.
    Note that this is pure $\K$-multiplication of a Toeplitz matrix with a vector.
    That can be accomplished in something like $O(\param \log \param)$.
    Afterwards, computing $\ve{\lambda_i} D_\theta \vec v$, for the computed column vector $\vec v$, can be done trivially in $O(\param)$.
    \\
    On a meta-level, I find this very interesting: fast commutative arithmetic can be applied by pushing the skewness of the ring into precomputation ($\Xi_{g_h}$) as well as after the bottleneck computation. Could a similar strategy be applied elsewhere?
  \end{itemize}
}

Numerous recent publications have dealt with shaping the core of various decoding algorithms for Reed--Solomon (RS) and other codes around $\F_{q}[x]$ module minimisation, lattice basis reduction or module Gr\"obner basis computation: three computational concepts which all converge to the same in this instance.
First for the Guruswami--Sudan list decoder \cite{alekhnovich_linear_2005,lee_list_2008,cohn_ideal_2010}, then for Power decoding \cite{nielsen2013generalised} and also either type of decoder for Hermitian codes \cite{nielsen_sub-quadratic_2014}. 

The impact of this can be said to be two-fold: firstly, by factoring out coding theory from the core problem, we enable the immediate use of sophisticated algorithms developed by the computer algebra community such as \cite{giorgi_complexity_2003,zhou_efficient_2012}.
Secondly, the setup has proved very flexible and readily applicable in settings which were not envisioned to begin with, such as the aforementioned Power decoder for Hermitian codes, or recently for Power decoding of RS codes up to the Johnson bound \cite{nielsen2014power}.

The main goal of this paper is to extend the module minimisation description to skew polynomial rings and Gabidulin codes, in particular Interleaved Gabidulin codes, with the aim of enjoying similar benefits.
Concretely, we lay a foundation by extending the core terms of weak Popov form and orthogonality defect, as well as extending the elegantly simple Mulders--Storjohann algorithm \cite{mulders2003lattice} to matrices over skew polynomial rings.
We analyse its complexity when applied to the shift register problem which arise when decoding Interleaved Gabidulin codes.
Finally, we extend the Demand--Driven algorithm for $\F_q[x]$ shift register problems \cite{nielsen2013generalised}, which is derived from the Mulders--Storjohann, also to the skew polynomial setting.

Gabidulin codes \cite{delsarte1978bilinear,gabidulin1985theory,roth1991maximum} are maximum rank distance codes with various applications like random linear network coding \cite{koetter2008coding,silva2008rank} and cryptography \cite{gabidulin1991ideals}.
They are the rank-metric analogue of RS codes.
An Interleaved Gabidulin code \cite{loidreau2006decoding,silva2008rank,sidorenko2011skew} is a direct sum of several $(n,k_i)$ Gabidulin codes: these can be decoded in a collaborative manner, improving the error-correction capability beyond the usual half the minimum rank distance of Gabidulin codes.
Similar to Interleaved RS codes, see \cite{schmidt2009collaborative} and its references, the core task of decoding can be reduced to what is known as a multi-sequence skew-feedback shift register synthesis problem \cite{sidorenko2011skew}.

In this paper, we use the introduced module minimisation description to solve a more general form of this problem, which we abbreviate MgLSSR:
\begin{problem}[MgLSSR]\label{prob:modmin}
Given skew polynomials $s_i, g_i$ and non-negative integers $\gamma_i \in \N_0$ for $i=1,\ldots,\ell$, find skew polynomials $\lambda, \omega_1,\dots,\omega_\ell$, with $\lambda$ of minimal degree such that the following holds:
\begin{eqnarray}
&\lambda s_i \equiv \omega_i \mod g_i &  \label{eq:congr} \\
&\deg\omega_i + \gamma_i < \deg \lambda + \gamma_0 &	\label{eq:deg}
\end{eqnarray}
\end{problem}
The original problem of \cite{sidorenko2011skew} set $g_i$ to powers of $x$ and $\gamma_i = 0$.
The above is a natural generalisation, which covers error and erasure decoding of Gabidulin codes \cite{LSS2014}, as well as an Interleaved extension of the Gao-type decoder for Gabidulin codes (\!\!\cite[\S 3.2]{wachter2013interpolation} combined with the ideas of \cite{sidorenko2011skew}).
For cases where the algorithm of \cite{sidorenko2011skew} applies, the Demand--Driven algorithm we present has the same complexity.
However, the more general perspective of module minimisation gives conceptually simpler proofs, and may prove useful for gaining further insights or faster, more sophisticated algorithms.

Normal form computation of matrices over skew rings and Ore rings has been investigated before, e.g.~\cite{abramov_solutions_2001,beckermann2006fraction}, but the focus has been over rings such as $\Z$ or $\K[z]$ for some field $\K$, where coefficient growth is important to control.
Since we are inspired mainly by the application to Gabidulin codes, where the skew ring is over a finite field, we count only operations performed in the field; in this measure those previous algorithms are much slower than what is presented here.

We set basic notation in \cref{sec:notations}.
\cref{sec:modulemin} describes how to solve \cref{prob:modmin} using module minimisation, and gives the Mulders--Storjohann algorithm for skew polynomial modules to accomplish this.
We introduce important concepts for arguing about such modules in \cref{sec:complexity} for performing a complexity analysis.
\cref{sec:demanddriven} describes how to then derive the faster Demand--Driven algorithm.
Due to lack of space, a number of proofs are omitted.

\section{Notation and Remarks on Generality}\label{sec:notations}
Let $\K$ be a field.
Denote by $\R=\K[x;\theta,\delta]$ the noncommutative ring of skew polynomials over $\K$ with automorphism $\theta$ and derivation $\delta$.
Being an Ore extension, $\R$ is both a left and right Euclidean ring.
See \cite{ore1933theory} for more details.

For coding theory we usually take $\K$ as a finite field $\F=\F_{q^r}$ for a prime power $q$ and $\theta$ as the Frobenius automorphism $\theta(a)=a^{q}$ for $a\in \F_{q^r}$.
Also, non-vanishing derivations $\delta$ are usually not considered, a notable exception being \cite{boucher2014linear}.
The algorithms in this paper are correct for any field, automorphism and derivation.
For complexities, we are counting field operations, and we often assume $\delta = 0$.

By $a \equiv b \mod c$  we denote the \emph{right} modulo operation in $\R$, i.e., that there exists $d \in \R$ such that $a = b + d c$.
By ``modules'' we will mean left $\R$-modules.
We extensively deal with vectors and matrices over $\R$.
Matrices are named by capital letters (e.g. $V$). The $i$th row of $V$ is denoted by $\v_i$ and the $j$th element of a vector $\v$ is $v_j$.
$v_{ij}$ is the $(i,j)$th entry of a matrix $V$. Indices start at $0$.
\begin{itemize}\setlength\itemsep{.3em}
\item The \emph{degree of a vector} $\v$ is $\deg \v := \max_{i} \{ \deg v_i \}$ (and $\deg \ve 0 = -\infty$) and the \emph{degree of a matrix} $V$ is $\deg V := \sum_{i} \{ \deg \v_i \}$.
\item The \emph{max-degree} of $V$ is $\maxdeg V := \max_{i} \{ \deg \v_i \} = \max_{i,j} \{ \deg v_{ij} \}$.
\item The \emph{leading position} of a vector $\v$ is $\LP{\v} := \max\{i \, : \, \deg v_i = \deg \v \}$.
      Furthermore $\LT{\v} := v_{\LP{\v}}$ and $\LC \v$ is the leading coefficient of $\LT \v$.
\end{itemize}

\begin{inJournal}
We say that vectors $\v_1,\v_2,\dots,\v_n \in \R^{n}$ are (left) linearly independent over $\R$ if from a vanishing linear combination $a_1 \v_1+a_2 \v_2,\dots,a_n \v_n=0$ it follows that $a_1 = a_2 = \dots = a_n=0$.
Using this notation, we can define the (left) rank of a matrix as the maximum number of rows which are linearly independent over $\R$.
A basis of a module is a set of vectors (e.g., in matrix form) which (left) span the module and are linearly independent.
It can be shown that by performing elementary row operations on a basis we obtain again a basis of the same module \cite{beckermann2006fraction}.
Modules having a basis are called free modules.
\end{inJournal}

\section{Finding a Solution using Module Minimisation}\label{sec:modulemin}

In the sequel we consider a particular instance of \cref{prob:modmin}, so $\R$, $\ell \in \N$, and $s_i, g_i \in \R$, $\gamma_i \in \N_0$ for $i=1,\ldots,\ell$ are arbitrary but fixed.
We assume $\deg s_i \leq \deg g_i$ for all $i$ since taking $s_i := s_i \modop g_i$ yields the same solutions.
Denote by $\M$ the set of all vectors $\v \in \R^{\ell+1}$ satisfying \eqref{eq:congr}, i.e.,
\begin{equation}\label{eq:module_vectors}
\M := \big\{ (\lambda,\omega_1,\dots,\omega_\ell)\in \R^{\ell+1} \; | \; \lambda s_i \equiv \omega_i \mod g_i \; \forall i=1,\dots,\ell \}.
\end{equation}
\vspace*{-1.8em}
\begin{inConference}
\begin{lemma}\label{LemM}\label{lem:basis}
    $\M$ with component-wise addition and left multiplication by elements of $\R$ forms a left module over $\R$.
    The rows of $M$ form a basis of $\M$:
    \renewcommand{\arraystretch}{0.9}
    \[M=\left (
    \begin{array}{ccccc}
    1 & s_1 & s_2 & \dots & s_\ell\\
    0 & g_1 & 0 & \dots  & 0 \\
    0 & 0 & g_2 & \dots  & 0 \\
    \vdots & \vdots & \ddots & \ddots&\vdots \\
    0 & 0 & 0 & \dots  & g_\ell \\
    \end{array}\right)
  \]
\end{lemma}
\end{inConference}
\begin{inJournal}
\begin{lemma}\label{LemM}
The set $\M$ with component-wise addition and left multiplication by elements of $\R$ forms a left module over $\R$.
\end{lemma}
\begin{proof}
Since the direct sum $\R^{\ell+1}$ is a left module over $\R$ and $\M \subset \R^{\ell+1}$, it suffices to show that $\M$ is closed under addition and left multiplication. If $\v \in \M$, it follows from \eqref{eq:module_vectors} that
\begin{equation}\label{eq:congr2}
v_0 s_i = v_i+a_i g_i \quad \forall i \in \{1, \dots, \ell\}
\end{equation}
for some $a_i\in\R$. If $\v,\v'\in \M$ then \eqref{eq:congr2} implies that  $(v_0+v_0') s_i = (v_i+ v'_i)+(a_i+a'_i) g_i$,  $i=1,\dots,\ell$, which means that $\v + \v' \in \M$. If $\v \in \M$ and $b \in \R$ then from \eqref{eq:congr2} we get $b v_0 s_i = b v_i + b a_i g_i$ and thus, $b \v \in \M$.
\end{proof}

\begin{lemma}\label{lem:basis}
The rows of the following matrix $M$ form a basis of $\M$:
\begin{equation}\label{eq:basis}
M=\left (
\begin{array}{ccccc}
1 & s_1 & s_2 & \dots & s_\ell\\
0 & g_1 & 0 & \dots  & 0 \\
0 & 0 & g_2 & \dots  & 0 \\
\vdots & \vdots & \ddots & \ddots&\vdots \\
0 & 0 & 0 & \dots  & g_\ell \\
\end{array}\right)
\end{equation}
\end{lemma}

\begin{proof}
It follows from \eqref{eq:module_vectors} that $\v = (v_0,\dots,v_\ell)\in\M$ $\Leftrightarrow$
$\exists a_1,\dots,a_\ell \in \R$ such that
$v_0 s_i = v_i+a_i g_i \quad \forall i \in \{1, \dots, \ell\}$
Therefore
\begin{equation*}\label{eq:linear_comb}
\v=v_0 \u_0 - a_1 \u_1 -\dots-a_\ell \u_\ell \ ,
\end{equation*}
where $\u_0,\dots,\u_\ell$ are the rows of $M$.
Thus, $\v$ is in the (left) span of $M$.
The rows of $M$ are linearly independent since $M$ is in triangular form.
\end{proof}
\end{inJournal}

\begin{inJournal}
  This proves that $\M$ is a free module.
  More generally, we show in \cref{sec:appendix_IBN} that any submodule of $\R^{n}$ is free and two finite bases of the same $\R$-module have the same number of elements.
\end{inJournal}

The above gives a simple description of all solutions of the congruence relations \eqref{eq:congr}.
To solve \cref{prob:modmin}, we therefore need an element in the $\M$ which satisfies the degree condition \eqref{eq:deg} and has minimal degree.
For this purpose, define
\begin{equation}
\Phi: \R^{\ell+1} \to \R^{\ell+1}, \quad \u = (u_0,\dots,u_\ell) \mapsto (u_0 x^{\gamma_0},\dots,u_\ell x^{\gamma_\ell}).
\end{equation}
We can extend the domain of $\Phi$ to matrices over $\R$ by applying it row-wise.
It is easy to see that $\Phi(\M)$ is also a left $\R$-module and that $\Phi$ is a module isomorphism.
Using this notation, we can restate how to solve \cref{prob:modmin}:
\begin{lemma}\label{lem:popov}
  A vector $\v \in \M^*$ is a solution to \cref{prob:modmin} if and only if $LP(\Phi(\v))=0$ and for all $\ve{u} \in \M^*$ with $LP(\Phi(\ve{u}))=0$ it holds that $\deg \Phi(\v) \le \deg \Phi(\ve{u})$.
\end{lemma}
\begin{proof}
$\v \in \M^*$ is a solution to \cref{prob:modmin} iff it satisfies \eqref{eq:deg} and $v_0$ has minimum possible degree.
That $\v$ satisfies \eqref{eq:deg} means $\deg v_0 + \gamma_0 > \deg v_i + \gamma_i$ and so $\deg (v_0 x^{\gamma_0}) > \deg (v_i x^{\gamma_i})$ i.e.~$\LP{\Phi(\v)}=0$.
The reverse direction is similar.
\end{proof}

So we should find a vector $\v \in \Phi(\M)$ with minimum-degree leading term among vectors with leading position zero.
We do this by finding a basis of $\Phi(\M)$ of a specific form.
This extends similar ideas for matrices over $\K[x]$ \cite{mulders2003lattice,nielsen2013generalised}.
\begin{definition}
A matrix $V$ over $\R$ is in weak Popov form if the leading positions of  all its non-zero rows are different.
\end{definition}

The following \emph{value function} for $\R$ vectors will prove useful: $\vecVal: \, \R^{\ell+1} \to \N_0$,
\[
  \vecVal(\v)=(\ell+1)\deg \v + \LP{\v} + 1 \qquad \textrm{for } \v \neq \vec 0 \textrm{ and } \qquad \vecVal(\vec 0) = 0.
\]

\begin{lemma}\label{lem:minpop}
Let $V$ be a matrix in weak Popov form whose rows are a basis of a left $\R$-module $\V$.
Then every $\u \in \V^*$ satisfies $\deg \u \ge \deg \v$, where $\v$ is the row of $V$ with $\LP{\v}=\LP{\u}$.
\end{lemma}
\begin{proof}
Let $\u \in \V^*$, and so $\exists a_0,\dots,a_\ell \in \R$ s.t.~$\u = \sum_{i=0}^{\ell} a_i \v_i$.
The $\u_i$ all have different leading position, so the $a_i\v_i$ must as well for those $a_i \neq 0$, which in turn means that the their $\vecVal(a_i\v_i)$ are all different.
Notice that for any two $\u_1, \u_2$ with $\vecVal(u_1) \neq \vecVal(u_2)$, then $\vecVal(\u_1+\u_2)$ either equals $\vecVal(\u_1)$ or $\vecVal(\vec u_2)$.
Applied inductively, that implies that there is an $i$ such that $\vecVal(\u) = \vecVal(a_i\v_i)$, which gives $\LP \u = \LP{\v_i}$ and $\deg \u = \deg a_i + \deg \v_i$.
\end{proof}

\cref{lem:popov} and \cref{lem:minpop} imply that a basis of $\M$ in weak Popov form gives a solution to \cref{prob:modmin} as one of its rows.
The following definition leads to a remarkably simple algorithm for computing such a basis: \cref{alg:ms}, an $\R$ variant of the Mulders--Storjohann algorithm \cite{mulders2003lattice}, originally described for $\K[x]$.

\begin{definition}\label{def:simple_transformation}
  Applying a simple transformation on a matrix $V$ means finding non-zero rows $\v_i, \v_j$, $i \neq j$ such that $\LP{\v_i}= \LP{\v_j}$ and $\deg \v_i \le \deg \v_j$, and replace $\v_j$ by $\v_j - \alpha x^\beta \v_i$,
  where $\beta = \deg v_j - \deg v_i$ and $\alpha = \LC{\v_j}/\theta^{\beta}(\LC{\v_i})$.
\end{definition}

\begin{remark}
  Note that a simple transformations cancels the leading term of the polynomial $\LT{v_j}$.
  Also elementary row operations keep the module spanned by the matrix' rows unchanged, see e.g.~\cite{beckermann2006fraction}, so the same is true for any sequence of simple transformations.
\end{remark}

\begin{lemma}\label{lem:val_dec}
If $\v'$ replaces $\v$ in a simple transformation, then $\vecVal(\v')<\vecVal(\v)$.
\end{lemma}
\begin{proof}
The operations used in a simple transformation ensure that $\deg \v' \leq \deg \v$.
If $\deg \v' < \deg \v$, we are done because $\LP{\v'} < \ell+1$.
If $\deg \v' = \deg \v$, then $\LP{\v'} < \LP{\v}$: by the definition of the leading position, all terms to the right of $\LP{\v}$ in $\v$ and $\alpha x^\beta \v_i$, and therefore also in $\v'$, have degree less than $\deg \v$.
Furthermore $\deg \v'_{\LP{\v}} < \deg \v$ by the definition of a simple transformation.
\end{proof}
\vspace{-0.5cm}
\begin{algorithm}[H]
  \caption{Mulders--Storjohann for $\R$ matrices}
  \label{alg:ms}
  \begin{algorithmic}[1]
    \Require{A square matrix $V$ over $\R$, whose rows span the module $\V$}
    \Ensure{A basis of $\V$ in weak Popov form.}
    \State Apply simple transformations on the rows of $V$ until no longer possible.
    \State \Return $V$.
  \end{algorithmic}
\end{algorithm}

\begin{theorem}
\cref{alg:ms} is correct.
\end{theorem}
\begin{proof}
  By \cref{lem:val_dec}, the value of one row of $V$ decreases for each simple transformation.
  The sum of the values of the rows must at all times be non-negative so the algorithm must terminate.
  Finally, when the algorithm terminates there are no simple transformations possible on $V$ anymore, i.e.~there are no $i \neq j$ such that $\LP{\v_i}= \LP{\v_j}$.
  That is to say, $V$ is in weak Popov form.
\end{proof}

This gives an algorithm to solve \cref{prob:modmin}.
The above proof could also easily lead to a rough complexity estimate.
To obtain a more fine-grained one, we will in the next section restrict ourselves to matrices which are square and full rank.

\section{Complexity Analysis}\label{sec:complexity}

Lenstra \cite{lenstra85factoring} introduced the notion of orthogonality defect of square, full rank $\K[x]$ matrices, and in \cite{nielsen2013generalised}, it was shown it can describe the complexity of the Mulders--Storjohann and Alekhnovich \cite{alekhnovich_linear_2005} algorithms for such matrices more fine-grained than originally, and that this improves the asymptotic estimate when the input comes from shift register problems.
The same concept cannot immediately be carried over to $\R$ matrices, since it is defined using the determinant.
For noncommutative rings, there are no functions behaving exactly like the classical determinant, but the \Dieu determinant \cite{dieudonne1943det} shares sufficiently many properties with it for our use.
Simply defining this determinant requires us to pass to the field of fractions of $\R$.

\subsection{Dieudonn\'e Determinant and Orthogonality Defect}

The following algebra is standard for noncommutative rings, so we will go through it quickly; more details can be found in \cite[Chapter 1]{cohn1977skew}.
We know that $\R$ is a principal left ideal domain which implies that it is left Ore and therefore has a unique left field of fractions $\Q = \{s^{-1}r \, : \, r \in \R, s \in \R^*\} /(\sim)$, where $\sim$ is the congruence relation $s^{-1}r  \sim s'^{-1}r' $ if $\exists u,u' \in \R^*$ such that $u r = u' r'$ and $u s = u' s'$.
The degree map on $\R$ can be naturally extended to $\Q$ by defining
\begin{equation*}
\deg: \, \Q \to \mathbb{Z} \cup \{-\infty\}, \quad s^{-1} r \mapsto \deg r - \deg s.
\end{equation*}

Let $[\Q^*,\Q^*]$ be the commutator of $\Q^*$, i.e.~the multiplicative group generated by $\{a^{-1}b^{-1}ab \, : \, a,b \in \Q^*\}$.
Then $\Qab = \Q^*/[\Q^*,\Q^*]$ is an abelian group called the multiplicative abelianization of $\Q^*$.
There is a canonical homomorphism
\begin{equation*}
\phi: \, \Q^* \to \Qab, \quad x \mapsto x \cdot [\Q^*,\Q^*].
\end{equation*}
Since the elements $(a^{-1}b^{-1}ab) \in [\Q^*,\Q^*]$ have degree $\deg(a^{-1}b^{-1}ab) = \deg (ab) - \deg (ba) = 0$, we can pass $\deg$ through $\phi$ in a well-defined manner: $\deg \phi(x) = \deg x$ for all $x \in \Q^*$.
The following lemma was proved by Dieudonn\'e \cite{dieudonne1943det} and can also be found in \cite{draxl1983skew}.

\begin{lemma}\label{lem:detexists}
There is a function $\det:\,\MQ \to \Qab$ s.t. for all $A \in \MQ$, $k \in \Q$:
{\begin{enumerate}
    \setlength{\itemsep}{0.2em}
\item[(i)] $\det I = 1$, where $I$ is the identity matrix in $\MQ$.
\item[(ii)] If $A'$ is obtained from $A$ by an elementary row operation, then $\det A' = \det A$.
\item[(iii)] If $A'$ is obtained from A by multiplying a row with $k$, then $\det A' = \phi(k) \det A$.
\end{enumerate}
}
\end{lemma}

\begin{definition}
  A function $\det$ with the properties of \cref{lem:detexists} is called a \Dieu determinant.
\end{definition}

Note that contrary to the classical determinant, a \Dieu determinant is generally not unique.
For the remainder of the paper, consider $\det$ to be any given \Dieu determinant.

\begin{lemma}\label{lem:dettriangular}
Let $A \in \MQ$ be in triangular form with non-zero diagonal elements $d_0,\dots,d_{n-1}$.
Then $\det A = \prod_{i=0}^{n-1} \phi(d_i)$.
\end{lemma}
\begin{proof}
  Since $d_i \neq 0$ for all $i$, we can multiply the $i$th row of $A$ by $d_i^{-1}$ and get a unipotent triangular matrix $A'$.
  Any unipotent triangular matrix can be obtained by elementary row operations from the identity matrix $I$. Thus
\begin{equation*}
\det A \overset{\substack{\text{Lemma} \\ \text{\ref{lem:detexists} (iii)}}}{=} \left[ \prod_{i=0}^{n-1} \phi(d_i) \right] \cdot \det A' 
\overset{\substack{\text{Lemma} \\ \text{\ref{lem:detexists} (ii)}}}{=} \left[ \prod_{i=0}^{n-1} \phi(d_i) \right] \cdot \det I 
\overset{\substack{\text{Lemma} \\ \text{\ref{lem:detexists} (i)}}}{=} \prod_{i=0}^{n-1} \phi(d_i).
\end{equation*}
\end{proof}

Clearly, the notion of weak Popov form generalises readily to matrices over $\Q$.
We will now examine how this notion interacts with the \Dieu determinant and introduce the concept of orthogonality defect.
The statements in this section are all $\Q$ variants of the corresponding statements for $\K[x]$ matrices, see \cite{nielsen2013generalised}.
\begin{definition}
The orthogonality defect of $V$ is $\OD{V} := \deg V - \deg \det V$.
\end{definition}

\begin{lemma}\label{lem:ODzero}
If $V \in \GLQ$ is in weak Popov form, then $\OD{V} = 0$.
\end{lemma}
\begin{proof}[proof sketch]
We can assume that $\LP{v_i}=i$ for all $i$ because if not, we can change the order of the rows of $V$ and obtain a matrix with the same determinant and degree.
We can then apply elementary row operations to bring the matrix to upper triangular form.
After these row operations, the property $\LP{v_i}=i$ is preserved and $\deg v_{ii}$ is equal to $\deg v_{ii}$ of the start matrix for all $i$.
By \cref{lem:dettriangular} the degree of the determinant equals the sum of the degree of the diagonal elements, and hence $\deg V = \deg \det V$.
\end{proof}

\begin{inJournal}
\begin{proof}
Let $V \in \GLQ$.
It is obvious that $\LT{\v_i} \neq 0$ because otherwise $\v_i = (0,\dots,0)$ and $V$ would not be invertible.
We can assume w.l.o.g that $\LP{v_i}=i$ because if not, we can change the order of the rows (elementary row operation) of $V$ and obtain a matrix with the same determinant and degree.
We consider a sequence of matrices $V^{(k)}$ for $k=0,\dots,n-1$, where $V^{(0)} = V$ and $V^{(k+1)}$ is obtained from $V^{(k)}$ by the following elementary row operations:
\begin{equation}
\v_i^{(k+1)} = \v_i^{(k)}-v_{ik}^{(k)} \big(v_{kk}^{(k)}\big)^{-1} \v_k^{(k)} \quad \forall i=k+1,\dots,n-1
\end{equation}
We can prove the following properties for any $k=0,\dots,n-1$ by induction.
\begin{eqnarray}
\LP{\v_i^{(k)}} =& i &\forall i=0,\dots,n-1 \label{eq:c1} \\
\deg v_{ii}^{(k)} =& \deg v_{ii} &\forall i=0,\dots,n-1 \label{eq:c2} \\
v_{ij}^{(k)} =& 0 &\forall j<i \land j<k \label{eq:c3}
\end{eqnarray}
For $k=0$, these properties are fulfilled because $V^{(0)} = V$. Now suppose \eqref{eq:c1} - \eqref{eq:c3} are true for some $k \in \{0,\dots,n-2\}$.
Then \eqref{eq:c3} follows for $k+1$ because
\begin{equation*}
v_{ik}^{(k+1)} = v_{ik}^{(k)} - v_{ik}^{(k)} \big(v_{kk}^{(k)}\big)^{-1} v_{kk}^{(k)} = v_{ik}^{(k)} - v_{ik}^{(k)} = 0 \quad \forall i>k.
\end{equation*}
Due to $\LP{\v_k} = k$, $\deg v_{kj}^{(k)} < \deg v_{kk}^{(k)}$ for $j>k$ and it follows that
\begin{equation*}
\deg \big(v_{ik}^{(k)} \cdot (v_{kk}^{(k)})^{-1} \cdot v_{kj}^{(k)} \big) = \deg v_{ik}^{(k)} + \deg v_{kj}^{(k)} - \deg v_{kk}^{(k)} < \deg v_{ik}^{(k)} \overset{\LP{v_i} = i}{\leq} \deg v_{ii}^{(k)}.
\end{equation*}
Hence, for $j>k$ it holds that
\begin{equation*}
\deg v_{ij}^{(k+1)} = 
\begin{cases}
-\infty < \deg v_{ii} = \deg v_{ii}^{(k+1)} \quad (\text{since } v_{ij}^{(k+1)} = 0), &\text{ if } j\leq k\\[2ex]
\deg \big( \underset{\deg \leq \deg v_{ii}^{(k)}}{ \underbrace{v_{ij}^{(k)}}}- \underset{\deg < \deg v_{ii}^{(k)}}{\underbrace{v_{ik}^{(k)} \cdot (v_{kk}^{(k)})^{-1} \cdot v_{kj}^{(k)}}} \big) \leq \deg v_{ii}^{(k)} = \deg v_{ii}^{(k+1)}, &\text{ if } k<j<i\\[6ex]
\deg \big( v_{ii}^{(k)}- v_{ik}^{(k)} \cdot (v_{kk}^{(k)})^{-1} \cdot v_{ki}^{(k)} \big) = \deg v_{ii}^{(k)} = \deg v_{ii}, &\text{ if } j=i\\[2ex]
\deg \big( \underset{\deg < \deg v_{ii}^{(k)}}{ \underbrace{v_{ij}^{(k)}}}- a_{ik}^{(k)} \cdot (v_{kk}^{(k)})^{-1} \cdot v_{kj}^{(k)} \big) < \deg v_{ii}^{(k)} = \deg v_{ii}^{(k+1)}, &\text{ if } i<j
\end{cases}
\end{equation*}
This implies that $\LP{\v_i} = i$ \eqref{eq:c1} and $\deg v_{ii}^{(k)} = \deg v_{ii}$ \eqref{eq:c2} for $k+1$. Hence, $V^{(n-1)}$ is in upper triangular form, its diagonal elements have degree $\deg v_{ii}^{(n-1)} = \deg v_{ii}$ and by \cref{lem:dettriangular}
\begin{equation*}
\deg \det V^{(n-1)} = \deg \det \big( \prod_{k=0}^{n-1} \phi\big(v_{ii}^{(n-1)}\big) \big) = \sum_{i=0}^{n-1} \deg \phi\big(v_{ii}^{(n-1)}\big) = \sum_{i=0}^{n-1} \deg v_{ii}^{(n-1)} = \sum_{i=0}^{n-1} \deg v_{ii}.
\end{equation*}
Since $V^{(n-1)}$ is obtained from $V$ by elementary row operations, it follows that
\begin{equation*}
\OD{V} = \deg V - \deg \det V = \sum_{k=0}^{n-1} \deg \v_k - \deg \det V^{(n-1)} = \sum_{k=0}^{n-1} \deg v_{ii}  - \sum_{i=0}^{n-1} \deg v_{ii} = 0
\end{equation*}
\end{proof}
\end{inJournal}

\subsection{Complexity of Mulders--Storjohann}

We can now bound the complexity of \cref{alg:ms} using arguments similar to those in \cite{nielsen2013generalised}.
These are in turn, the original arguments of \cite{mulders2003lattice} but finer grained by using the orthogonality defect.
In the following, let $\param := \max_i \{\gamma_i +\deg g_i\}$.
We can assume that $\gamma_0 < \param$ since otherwise $(1, s_1, \ldots, s_\ell)$ is the minimal solution to the MgLSSR.

\begin{lemma}\label{lem:ODM}
$\OD{\Phi(M)} \leq \param-\gamma_0$.
\end{lemma}
\begin{inJournal}
\begin{proof}
Since $\Phi(M)$ is in upper triangular form, the degree of its determinant is
\begin{equation*}
  \textstyle
\deg \det \Phi(M) = \deg \left(\phi(x^{\gamma_0}) \prod_{i=0}^{\ell} \phi(g_i x^{\gamma_i}) \right) = \gamma_0 + \sum_{i=0}^{\ell}(\gamma_i + \deg g_i).
\end{equation*}
The degrees of the rows of $\Phi(M)$ are
\begin{equation*}
\deg \Phi(\m_0) = \max_i \{\gamma_i + \deg s_i\}
\quad \text{and} \quad
\deg \Phi(\m_i) = \gamma_i + \deg g_i \quad \text{for } i \geq 1.
\end{equation*}
Thus, 
$ 
\OD{\Phi(M)} = \max_i \{\gamma_i +\deg s_i\} - \gamma_0 \leq \param-\gamma_0.
$ 
\end{proof}
\end{inJournal}

\begin{theorem}
  \label{thm:ms_complexity}
Over $\R$ with derivation zero, \cref{alg:ms} with input matrix $\Phi(M)$ performs at most $(\ell+1) (\param-\gamma_0+1)$ simple transformations and performs $O(\ell^2 \param^2)$ operations over $\K$.
\end{theorem}
\begin{proof}
Every simple transformation reduces the value $\vecVal$ of one row with at least 1.
So the number of possible simple transformations is upper bounded by the difference of the sums of the values of the input matrix $\Phi(M)$ and the output matrix $V$, i.e.:
\begin{eqnarray*}
&&   {\textstyle \sum_{i=0}^{\ell}} [(\ell\!+\!1)\deg \Phi(\m_i)\!+\!\LP{\Phi(\m_i)}\!-\!\big( (\ell\!+\!1)\deg \Phi(\v_i)\!+\!\LP{\v_i} \big)] \\
&=& \LP{\Phi(\m_0)} + (\ell\!+\!1) {\textstyle \sum_{i=0}^{\ell}} [\deg \Phi(\m_i)\!-\!\deg \v_i]  \\
&\leq& (\ell\!+\!1) [\deg \Phi(M)\!-\!\deg V\!+\!1] = (\ell\!+\!1) [\OD{\Phi(M)}\!+\!1] ,
\end{eqnarray*}
where the last equality follows from $\deg V = \deg \det V = \deg \det M$.

One simple transformation consists of calculating $\v_j - \alpha x^\beta \v_i$, so for every coefficient in $\v_i$, we must apply $\theta^\beta$, multiply by $\alpha$ and then add it to a coefficient in $\v_j$, each being in $O(1)$.
Since $\deg v_j \leq \param$ this costs $O(\ell \param)$.
\end{proof}

\section{Demand-Driven Algorithm}\label{sec:demanddriven}

\def\p#1{\tilde #1}
\def\prev{\code{previous}}
\def\continue{\code{continue}}
\def\ass{\leftarrow}

It was observed in \cite{nielsen2013generalised} that the Mulders--Storjohann algorithm over $\K[x]$ admits a ``demand--driven'' variant when applied to matrices coming from shift register problems, where coefficients of the working matrix are computed only when they are needed.
This means a much lower memory requirement, as well as a better complexity under certain conditions.
Over $\R$, \cref{alg:ms} admits exactly the same speedup;
in fact, both the algorithm and the proof are almost line-for-line the same for $\R$ as for $\K[x]$.
We therefore focus on the idea of the algorithm, and the original proofs can be found in  \cite{nielsen2013generalised} (extended version).

The central observation is that due to the special form of $M$ of \cref{lem:basis}, only the first column is needed during the Mulders--Storjohann algorithm in order to construct the rest.
That is formalised in the following lemma:

\begin{lemma}
  \label{lem:msVar}
  Consider Algorithm \ref{alg:ms} with input $\Phi(M)$.
  Consider a variant where, when replacing $\vec v_j$ with $\vec v'_j$ in a simple transformation, instead replace it with $\vec v''_j = (v'_{j,0}, v'_{j,1} \modop \p g_1, \ldots, v'_{j,\ell} \modop \p g_\ell)$.
  This does not change correctness of the algorithm or the upper bound on the number of simple transformations performed.
\end{lemma}

The Demand--Driven algorithm, \cref{alg:dd}, therefore calculates just the first element of a vector whenever doing a simple transformation, being essentially enough information.
To retain speed it is important, however, that the algorithm can also figure out which simple transformation it can next apply, without having to recompute the whole matrix.
For this, we cache for each row its degree $\eta_j$ and the leading coefficient of its leading position $\alpha_j$.
The following observations then lead to \cref{alg:dd}:

\begin{enumerate}
  \item In $\Phi(M)$ there is at most one possible choice of the first simple transformation, due to the matrix' shape.
  This is true throughout the algorithm, making it deterministic.
  \item To begin with, if there is a possible simple transformation, row 0 is involved.
  Just before doing a simple transformation, we possibly swap the two rows involved such that the row changed is always row 0.
  That means row 0 is always involved if there is a possible simple transformation, and that the algorithm terminates when row 0 has leading position 0.
  \item To begin with row $i$ has leading position $i$ for $i>0$.
  The above swap ensures that this will keep being true.
  \item After doing a simple transformation, we need to update the degree, leading position and leading coefficient of only row 0; the rest remains unchanged.
  We do this by going through each possible degree and leading position in decreasing order of value $\vecVal$.
  This is correct since we know that the simple transformation must decrease the value of row 0.
\end{enumerate}

\begin{algorithm}[t]
  \caption{Demand--Driven algorithm for MgLSSR}
  \label{alg:dd}
  \begin{algorithmic}[1]
    \Require{$\p s_j \ass s_{1,j}x^{\gamma_j},\ \p g_j \ass g_jx^{\gamma_j}$ for $j=1,\ldots,\ell$\;}
    \Ensure{The first column of a basis of $\M$ whose $\Phi$ image is in weak Popov form. \;}
    \State $(\eta, h) \ass (\deg, \LPc)$ of $(x^{\gamma_0}, \p s_1, \ldots, \p s_\sigma)$
    \Ifline{$h = 0$}{\Return $(1, 0, \ldots, 0)$}
    \State $(\lambda_0,\ldots,\lambda_\ell) \ass (x^{\gamma_0},0,\ldots,0)$
    \State $\alpha_jx^{\eta_j} \ass $ the leading monomial of $\p g_j$ for $j=1,\ldots,\ell$
    \While{$\deg \lambda_0 \leq \eta$}
      \State $\alpha \ass $ coefficient to $x^\eta$ in $(\lambda_0\p s_h \mod \p g_h)$
      \label{line:dd_coefficient}
      \If{$\alpha \neq 0$}
      \label{line:dd_ifstatement}
        \Ifline{$\eta < \eta_h$}{swap $(\lambda_0,\alpha,\eta)$ and $(\lambda_h, \alpha_h, \eta_h)$}
        \State $\lambda_0 \ass \lambda_0 - \alpha/\theta^{\eta-\eta_h}(\alpha_h) x^{\eta-\eta_h} \lambda_h$
      \label{line:dd_transformation}
      \EndIf
      \State $(\eta, h) \ass  (\eta, h-1) \textbf{ if } h>1 \textbf{ else } (\eta-1, \ell)$
    \EndWhile
    \State \Return $\big(\lambda_0x^{-\eta_0},\ldots,\lambda_\ell x^{-\eta_0}\big)$
  \end{algorithmic}
\end{algorithm}

\def\supp{\word{supp}}
To express the complexity, by $\supp(f), f \in \R$, we mean the set of degrees such that $f$ has a non-zero coefficient for this degree.
By $\deg_2 f$ we mean the degree of the second largest coefficient.
Let again $\param := \max_i \{\gamma_i +\deg g_i\}$.

\begin{theorem}
  Algorithm \ref{alg:dd} is correct.
  Over $\R$ with derivation zero, it has computational complexity $O(\ell \param^2 \rho)$, where
  \[
    \rho = \left\{
      \begin{array}{l@{\quad}l}
        \max_i\{\#\supp(g_i)\} & \textrm{if } \deg_2 g_i < \frac{1}{2} \deg g_i \textrm{ for all } i \\
        \param                     & \textrm{otherwise}
      \end{array}\right.
  \]
  It has memory complexity $O(\ell \param)$.
\end{theorem}
\begin{proof}[proof sketch]
  We only prove the complexity statement.
  Clearly, all steps of the algorithm are essentially free except \cref{line:dd_coefficient} and \cref{line:dd_transformation}.
  Observe that every iteration of the while-loop decrease an \emph{estimate} on the value of row 0, whether we enter the if-branch in \cref{line:dd_ifstatement} or not.
  So by the arguments of the proof of \cref{thm:ms_complexity}, the loop will iterate at most $O(\ell \param)$ times.
  Each execution of \cref{line:dd_transformation} costs $O(\param)$ since the $\lambda_j$ all have degree at most $\param$.

  For \cref{line:dd_coefficient}, we can compute the needed coefficient $\alpha$ in complexity $O(\param \rho)$: if $\deg_2 g_h > \frac{1}{2}\deg g_h$, we simply compute the entire polynomial $\lambda_0\p s_h \mod \p g_h$ in time $O(\param^2)$.
  \begin{inConference}
  Otherwise, an easy argument shows that at most $\#\supp(g_h) + 1$ coefficients of $\lambda_j \p s_h$ affects the computation of $\alpha$.
  \end{inConference}
  \begin{inJournal}
  Otherwise, assume $\deg_2 g_h < \deg g_h/2$.
  We wish to describe the degrees of the coefficients which affect the computation of $\alpha$.
  Clearly, $\eta$ is one such.
  Then follows degrees $\eta + \deg g_h - i$ for any $i \in \supp(g_h)\setminus \{\deg g_h\}$ if the modulo reduction adds some scalar multiple of $x^{\eta-i}g_h$.
  Any coefficient influencing these will also influence $\alpha$, so transitively one sees that all degrees affecting $\alpha$ have the form
  \[
  \eta + t\deg g_h - i_1 - \cdots - i_t \qquad t \in \N_0,\ i_s \in \supp(g_h)\setminus \{\deg g_h\}
  \]
  (note that we here use that the derivation of $\R$ is zero).
  Of course, the above should not exceed $\deg(\lambda_0\p s_h) < \eta + \deg g_h$, so since $\deg_2 g_h < \deg g_h/2$ then $t \leq 1$.
  This means that at most $\#\supp(g_h)+1 \leq \rho$ coefficients affect the computation of $\alpha$.
  \end{inJournal}
  Each of these can be computed by convolution in time $O(\param)$.
\end{proof}

For generic $g_i$, \cref{alg:dd} will have complexity $O(\ell \param^3)$ which is usually worse than $O(\ell^2 \param^2)$ of \cref{alg:ms}.
However, for decoding of Interleaved Gabidulin codes, two important cases are $g_i = x^k$ (syndrome decoding \cite{sidorenko2011skew}) and $g_i = x^{q^m} - 1$ (Gao-type decoding \cite[\S 3.2]{wachter2013interpolation}), and here \cref{alg:dd} runs in complexity $O(\ell \param^2)$.

\begin{remark}
\cref{alg:dd} bears a striking similarity to the Berlekamp--Massey variant for multiple shift registers \cite{sidorenko2011skew} where all $g_i$ are powers of $x$, and has the same running time in this case.
However, using the language of modules, we obtain a more general algorithm with a conceptually simpler proof, and we can much more readily realise algebraic properties of the algorithm.
For instance, using known properties for the weak Popov form, it is trivial to prove that \cref{alg:dd} can be modified to return a basis for \emph{all} solutions to the shift register problem, as well as decompose any given solution as an $\R$-linear combination of this basis.
\end{remark}

\section{Conclusion}\label{sec:conclusion}

In this paper, we have given two module-based methods for solving generalised shift register problems over skew polynomial rings.
For ordinary polynomial rings, module minimisation has proven a useful strategy for obtaining numerous flexible, efficient while conceptually simple decoding algorithms for Reed--Solomon and other code families.
Our results introduce the methodology and tools aimed at bringing similar benefits to Gabidulin, Interleaved Gabidulin and other skew polynomial-based codes.

\subsubsection*{Acknowledgements}

Sven Puchinger (grant BO 867/29-3), Wenhui Li and Vladimir Sidorenko (both grant BO 867/34-1) were supported by the German Research Foundation ``Deutsche Forschungsgemeinschaft'' (DFG).
Johan S. R. Nielsen gratefully acknowledges the support of the Digiteo foundation, project \href{http://idealcodes.github.io/}{\color{navy}{IdealCodes}}.

\bibliographystyle{abbrv}
\bibliography{ModMinIG}

\begin{thebibliography}{10}

\bibitem{abramov_solutions_2001}
S.~A. Abramov and M.~Bronstein.
\newblock On solutions of linear functional systems.
\newblock In {\em Proc. of ISSAC}, pages 1--6, New York, {NY}, {USA}, 2001.

\bibitem{alekhnovich_linear_2005}
M.~Alekhnovich.
\newblock Linear {D}iophantine equations over polynomials and soft decoding of
  {R}eed–{S}olomon codes.
\newblock {\em {IEEE} Trans. Inf. Theory}, 51(7):2257--2265, July 2005.

\bibitem{beckermann2006fraction}
B.~Beckermann, H.~Cheng, and G.~Labahn.
\newblock Fraction-free row reduction of matrices of {O}re polynomials.
\newblock {\em J. Symb. Comp.}, 41(5):513--543, 2006.

\bibitem{boucher2014linear}
D.~Boucher and F.~Ulmer.
\newblock Linear codes using skew polynomials with automorphisms and
  derivations.
\newblock {\em Designs, Codes and Cryptography}, 70(3):405--431, 2014.

\bibitem{cohn_ideal_2010}
H.~Cohn and N.~Heninger.
\newblock Ideal forms of {C}oppersmith's theorem and {G}uruswami–{S}udan list
  decoding.
\newblock {\em {arXiv}}, 1008.1284, 2010.

\bibitem{cohn1977skew}
P.~M. Cohn.
\newblock {\em Skew Field Constructions}, volume~27.
\newblock Cambridge Univ. Press, 1977.

\bibitem{delsarte1978bilinear}
P.~Delsarte.
\newblock Bilinear forms over a finite field, with applications to coding
  theory.
\newblock {\em J. Comb. Th.}, 25(3):226--241, 1978.

\bibitem{dieudonne1943det}
J.~Dieudonn\'e.
\newblock Les d\'eterminants sur un corps non commutatif.
\newblock {\em Bull. Soc. Math. France}, 71:27--45, 1943.

\bibitem{draxl1983skew}
P.~K. Draxl.
\newblock {\em Skew Fields}, volume~81.
\newblock Cambridge Univ. Press, 1983.

\bibitem{gabidulin1985theory}
E.~M. Gabidulin.
\newblock Theory of codes with maximum rank distance.
\newblock {\em Problemy Peredachi Informatsii}, 21(1):3--16, 1985.

\bibitem{gabidulin1991ideals}
E.~M. Gabidulin, A.~Paramonov, and O.~Tretjakov.
\newblock Ideals over a non-commutative ring and their application in
  cryptology.
\newblock In {\em Eurocrypt}, pages 482--489, 1991.

\bibitem{giorgi_complexity_2003}
P.~Giorgi, C.~Jeannerod, and G.~Villard.
\newblock On the complexity of polynomial matrix computations.
\newblock In {\em Proc. of ISSAC}, pages 135--142, 2003.

\bibitem{koetter2008coding}
R.~Koetter and F.~R. Kschischang.
\newblock Coding for errors and erasures in random network coding.
\newblock {\em {IEEE} Trans. Inf. Theory}, 54(8):3579--3591, 2008.

\bibitem{lee_list_2008}
K.~Lee and M.~E. O'Sullivan.
\newblock List decoding of {R}eed–{S}olomon codes from a {G}röbner basis
  perspective.
\newblock {\em J. Symb. Comp.}, 43(9):645 -- 658, 2008.

\bibitem{lenstra85factoring}
A.~Lenstra.
\newblock Factoring multivariate polynomials over finite fields.
\newblock {\em J. Comp. Syst. Sc.}, 30(2):235--246, 1985.

\bibitem{LSS2014}
W.~Li, V.~Sidorenko, and D.~Silva.
\newblock On transform-domain error and erasure correction by {G}abidulin
  codes.
\newblock {\em Designs, Codes and Cryptography}, 73(2):571--586, 2014.

\bibitem{loidreau2006decoding}
P.~Loidreau and R.~Overbeck.
\newblock Decoding rank errors beyond the error correcting capability.
\newblock In {\em Proc. of ACCT}, pages 186--190, 2006.

\bibitem{mulders2003lattice}
T.~Mulders and A.~Storjohann.
\newblock On lattice reduction for polynomial matrices.
\newblock {\em J. Symb. Comp.}, 35(4):377--401, 2003.

\bibitem{nielsen2013generalised}
J.~S.~R. Nielsen.
\newblock Generalised multi-sequence shift-register synthesis using module
  minimisation.
\newblock In {\em Proc. of IEEE ISIT}, pages 882--886, 2013.
\newblock Extended version at
  \href{http://arxiv.org/abs/1301.6529}{http://arxiv.org/abs/1301.6529}.

\bibitem{nielsen2014power}
J.~S.~R. Nielsen.
\newblock Power decoding {R}eed--{S}olomon codes up to the {J}ohnson radius.
\newblock In {\em Proc. of ACCT}, 2014.

\bibitem{nielsen_sub-quadratic_2014}
J.~S.~R. Nielsen and P.~Beelen.
\newblock Sub-quadratic decoding of one-point {H}ermitian codes.
\newblock {\em {arXiv}}, 1405.6008, May 2014.
\newblock Submitted to {IEEE} Trans. Inf. Theory.

\bibitem{ore1933theory}
O.~Ore.
\newblock Theory of non-commutative polynomials.
\newblock {\em Annals of Mathematics}, 34(3):480--508, July 1933.

\bibitem{roth1991maximum}
R.~M. Roth.
\newblock Maximum-rank array codes and their application to crisscross error
  correction.
\newblock {\em {IEEE} Trans. Inf. Theory}, 37(2):328--336, 1991.

\bibitem{schmidt2009collaborative}
G.~Schmidt, V.~R. Sidorenko, and M.~Bossert.
\newblock Collaborative decoding of interleaved {R}eed--{S}olomon codes and
  concatenated code designs.
\newblock {\em {IEEE} Trans. Inf. Theory}, 55(7):2991--3012, 2009.

\bibitem{sidorenko2011skew}
V.~Sidorenko, L.~Jiang, and M.~Bossert.
\newblock Skew-feedback shift-register synthesis and decoding interleaved
  {G}abidulin codes.
\newblock {\em {IEEE} Trans. Inf. Theory}, 57(2):621--632, 2011.

\bibitem{silva2008rank}
D.~Silva, F.~R. Kschischang, and R.~Koetter.
\newblock A rank-metric approach to error control in random network coding.
\newblock {\em {IEEE} Trans. Inf. Theory}, 54(9):3951--3967, 2008.

\bibitem{wachter2013interpolation}
A.~Wachter-Zeh and A.~Zeh.
\newblock Interpolation-based decoding of interleaved {G}abidulin codes.
\newblock In {\em Proc. of WCC}, pages 528--538, 2013.

\bibitem{zhou_efficient_2012}
W.~Zhou and G.~Labahn.
\newblock Efficient algorithms for order basis computation.
\newblock {\em J. Symb. Comp.}, 47(7):793--819, July 2012.

\end{thebibliography}

\newpage

\appendix

\begin{inJournal}
\section{Proofs about free $\R$-modules}\label{sec:appendix_IBN}

Let $\R=\F[x;\theta,\delta]$ be the ring of skew polynomials over a finite field.

\vspace{0.5cm}

\begin{lemma}\label{lem:pid}
$\R$ is a principal ideal domain.
\end{lemma}

\begin{proof}
The claim follows directly from the fact that $\R$ is a Euclidean domain (due to the existence of a left- and right- division algorithm).
\end{proof}

\vspace{0.5cm}

\begin{theorem}
Let $F$ be a free left (right) module over $\R$ which has a finite basis $\{x_1, \dots, x_n\}$ for some $n \in \mathbb{N}$, and $M \subset F$ a left (right) submodule. Then $M$ is free and there exists a basis of $M$ which has $\leq n$ elements.
\end{theorem}

\begin{proof}
The idea of this proof is taken from \cite{lang2002algebra}. We prove the claim by constructing a basis of $F$ inductively. Before we start, we have to define the sets $M_r := M \cap (x_1, \dots, x_r)$ for any $r \in \mathbb{N}$ with $1 \leq r \leq n$. Note that we only show the claim for left modules. The right module case can be shown equivalently.

\vspace{0.2cm}
{\bf Base case:}

We first consider  $M_1 = M \cap (x_1)$. Since $M$ and $(x_1)$ are both left modules over $\R$, $M_1$ is also a left module over $\R$. Let $J$ be the set of all $a \in \R$ such that $a x_1 \in M$. Then $J$ is a left ideal of $\R$ because for any $a,b \in J$ and $r \in \R$ also $(a+b) \in J$ (because $(a+b) x_1 = a x_1 + b x_1 \in M$) and $(ra) \in J$ (because $(ra) x_1 = r (a x_1) \in M$). Due to \cref{lem:pid}, $J$ is (left) principal and therefore generated by some $a_1 \in \R$. Hence, $M_1$ must be of the form $(a_1 x_1)$. If $a_1 \neq 0$, $\{a_1 x_1\}$ is a basis of $M_1$. Otherwise, $M_1 = \{0\}$ and $\emptyset$ is a basis of $M_1$. Thus, $M_1$ is free and has a basis with $\leq 1$ elements.

\vspace{0.2cm}
{\bf Induction hypothesis:}

Now we assume that $M_r = M \cap (x_1, \dots, x_r)$ is free and has a basis $\{\tilde{x}_1, \dots, \tilde{x}_{\ell}\}$ with $0 \leq \ell \leq r$ for some $r$ with $1 \leq r < n$.

\vspace{0.2cm}
{\bf Inductive step:}

If $M_r = M_{r+1}$, we are done, because by hypothesis $M_{r+1}$ is free with the same basis as $M_r$ which also obviously has $\leq (r+1)$ elements. In the following we therefore assume that $M_{r+1} \setminus M_r \neq \emptyset$.

Let $I$ be the set of all $a \in \R$ such that there exists an element $x^{(a)} \in M$ and $b_1^{(a)}, \dots, b_r^{(a)} \in \R$ with $a x_{r+1} = x^{(a)} - \sum_{i=1}^{r}{b_i^{(a)} x_i}$. Then $I$ is a left ideal of $\R$ because if $a,b \in I$ and $r \in \R$, then $a+b \in I$ because $x^{(a+b)} = x^{(a)}+x^{(b)} \in M$ and $b_i^{(a+b)} = b_i^{(a)}+b_i^{(b)} \in \R$ and $r a \in I$ because $x^{(r a)} = r x^{(a)} \in M$ and $b_i^{(r a)} = r b_i^{(a)} \in \R$. Due to \cref{lem:pid}, $I$ is (left) principal and is therefore generated by some element $a_{r+1} \in \R$. Since we assume that $M_{r+1} \neq M_r$, $a_{r+1}$ must not be $0$. We can choose a $\gamma \in M_{r+1}$ such that $\gamma = x^{(a_{r+1})} \neq 0$.

\vspace{0.2cm}

{\it (Sub-)Claim: $B_{r+1} := \{\tilde{x}_1, \dots, \tilde{x}_{\ell}, \gamma\}$ is a basis of $M_{r+1}$.}

{\it (i) $B_{r+1}$ is a generating set of $M_{r+1}$:}

For any $x \in M_{r+1}$, the coefficient of $x$ with respect to $x_{r+1}$ is in $I$ and therefore divisible by $a_{r+1}$. Hence, there exists a $d \in \R$ such that the coefficient of $x - d \gamma$ with respect to $x_{r+1}$ is $0$ and therefore $x - d \gamma \in M_{r}$. This shows that $\{\tilde{x}_1, \dots, \tilde{x}_{\ell}, \gamma\}$ is a generating set of $M_{r+1}$.

{\it (ii) The elements of $B_{r+1}$ are linearly independent:}

Since $M_{r+1} \neq M_r$, there is an $x \in M_{r+1}$ which is not in $M_{r}$. Like in the previous statement, there is a $d \in \R$ such that $x - d \gamma \in M_{r}$ (note that here $d \neq 0$ because $x \notin M_r$). This shows that $\gamma \notin M_r$ because if it was, $d \gamma$ would be in $M_r$ (because $M_r$ is a module) and $x = (x - d \gamma) + d \gamma \in M_r$, which would be a contradiction. Thus, $\gamma \notin M_r = (\tilde{x}_1, \dots, \tilde{x}_{\ell})$ which implies that $\tilde{x}_1, \dots, \tilde{x}_{\ell}$ and $\gamma$ are linearly independent.

Hence, $B_{r+1}$ is a basis of $M_{r+1}$ which has $\ell+1 \leq r+1$ elements (because by hypothesis $\ell \leq r$).

\vspace{0.2cm}
{\bf Conclusion:}

At this point we found a basis $B_n$ for $M_n$ which has $\leq n$ elements. Since $\{x_1, \dots, x_n\}$ is a basis of $F$ and $M \subset F$, $M_n = M \cap (x_1, \dots, x_n) = M \cap F = M$ and therefore $B_n$ is also a basis of $M$. Note that the existence of a basis directly implies that $M$ is free.

\end{proof}

\vspace{0.5cm}

\begin{lemma}
$\R$ is (left- and right-) Noetherian. Hence, any finitely generated $\R$-module is Noetherian which implies that any submodule of a finitely generated $\R$-module is finitely generated.
\end{lemma}

\begin{proof}
A ring is (left-/right-) Noetherian if all its ideals are finitely generated, which is the case for $\R$ because it is (left-/right-) principal and therefore any (left-/right-) ideal is generated by one element. The second part of the claim follows from the definition of Noetherian rings.
\end{proof}

\vspace{0.5cm}

\begin{theorem}
Any two finite bases of a left (right) $\R$-module have the same number of elements.
\end{theorem}

\begin{proof}
Again we show the claim only for left modules, the right case can be shown equivalently. The ideas of this proof are taken from \cite{clark2012noncom}. Assume we have two bases $B_1$ and $B_2$ of the same $\R$-module with $|B_1| =: n \in \mathbb{N}$ and $|B_2| =: m \in \mathbb{N}$. We want to show that $n = m$.

We first note that from any $\R$-module with a basis $B = \{b_1, \dots, b_k\}$ of cardinality $k$, there is a natural module isomorphism $\varphi: (B) \rightarrow \R^k$, $\varphi(x) = (x_1, \dots, x_k)$ which is defined by the unique representation $x = \sum_{i=1}^{k} x_i b_i$ for any $x \in (B)$.

This means that $\R^n \cong (B_1) = (B_2) \cong R^n$, so $\R^n \cong R^m$, or in other words there is an isomorphism $\alpha: \R^n \rightarrow \R^m$.

We assume $n > m$. This means that $\alpha$ gives us an embedding $\alpha: \R^n = \R^m \bigoplus \R^{n-m} \hookrightarrow \R^m$ with $\R^{n-m} \neq \{0\}$ because $n>m$ and $\R \neq \{0\}$. Therefore, $A_0 := \R^m$ must have a submodule of the form $A_1 \bigoplus B_1$ such that $A_1 \cong \R^m$ and $B_1 \cong \R^{n-m} \neq \{0\}$. Constructing $A_i \bigoplus B_i$ as a submodule of $A_{i-1}$ inductively with the same approach, we obtain a submodule $\bigoplus_{i=1}^{\infty}B_i$ of $\R^m$ which is an infinite direct product of non-zero modules and can therefore not be finitely generated. This contradicts the fact that $\R^m$ is a (left) Noetherian module, which means that all its submodules are finitely generated. Hence, $n \leq m$. By the same argument we get that $m \leq n$ and therefore $n = m$.

\end{proof}

\section{Alternative Proofs}

{\bf Alternative proof of \cref{lem:bpr}}

\begin{proof}
Let $V'$ be the square matrix obtained from $V$ by replacing the $j$th row of $V$ by $v'_j = v_j - \alpha x^\delta v_i$ ($j \neq i$) and $v'_k = v_k$ for all $k \neq j$. Let $\varphi : \mathcal{R}^{\ell+1} \rightarrow \mathcal{R}^{\ell+1}, \varphi\left(a_0, \dots, a_{\ell}\right) := \left(\tilde{a}_0, \dots, \tilde{a}_{\ell}\right)$ with $\tilde{a}_k := \begin{cases}a_k, &\text{if } k \neq i \\ a_i - a_j \alpha x^\delta, &\text{if } k=i\end{cases}$. It is easy to construct the inverse function $\varphi^{-1}$. Thus, $\varphi$ is bijective. \\
{\it (Sub-)Claim}: For any $a = \left(a_0, \dots, a_{\ell}\right) \in \mathcal{R}^{\ell+1}$ and the corresponding $v \in \V$ such that $v = \sum_{k=0}^{\ell} a_k v'_k$ it holds that
\begin{IEEEeqnarray*}{rCl+c+rCl}
v &=& \sum_{k=0}^{\ell} \tilde{a}_k v_k &\iff& \tilde{a} &=& \varphi(a)
\end{IEEEeqnarray*}
{\it''$\Leftarrow$''}: Assume that $\tilde{a} = \varphi(a)$. Then
\begin{IEEEeqnarray*}{rCl}
v &=& \sum_{k=0}^{\ell} a_k v'_k \\
  &=& a_j v'_j + \sum_{k=0, k \neq j}^{\ell} a_k v_k \\
  &=& a_j (v_j - \alpha x^\delta v_i) + \sum_{k=0, k \neq j}^{\ell} a_k v_k \\
  &=& a_j v_j + (a_i - a_j \alpha x^\delta) v_i + \sum_{k=0, k \neq i, k \neq j}^{\ell} a_k v_k \\
  &=& \sum_{k=0}^{\ell} \tilde{a}_k v_k
\end{IEEEeqnarray*}
{\it''$\Rightarrow$''}: Assume that $v = \sum_{k=0}^{\ell} \tilde{a}_k v_k$. Then 
\begin{IEEEeqnarray*}{rClCl}
0 &=& v - v &=& \sum_{k=0}^{\ell} \tilde{a}_k v_k -\sum_{k=0}^{\ell} a_k v'_k \\
  &&&=& (\tilde{a}_i v_i - a_i v'_i) + (\tilde{a}_j v_j - a_j v'_j) + \sum_{k=0,k\neq i, k \neq j}^{\ell} (\tilde{a}_k v_k - a_k v'_k) \\
  &&&=& (\tilde{a}_i v_i - a_i v_i) + (\tilde{a}_j v_j - a_j (v_j - \alpha x^\delta v_i)) + \sum_{k=0,k\neq i, k \neq j}^{\ell} (\tilde{a}_k v_k - a_k v_k) \\
  &&&=& (\tilde{a}_i - (a_i - a_j \alpha x^\delta)) v_i + \sum_{k=0,k\neq i}^{\ell} (\tilde{a}_k - a_k) v_k
\end{IEEEeqnarray*}
Since $V$ is a basis of $\mathcal{V}$, its rows $v_k$ are linearly independent and therefore the coefficients of the sum with respect to each $v_k$ must be zero, which implies that $\tilde{a}_k =
\begin{cases}
a_k, &\text{if } k \neq i \\
a_i - a_j \alpha x^\delta, &\text{if } k = i
\end{cases}$, or in other words $\tilde{a} = \varphi(a)$.

Since $V$ is a basis of $\mathcal{V}$, for any $v \in \mathcal{M}$ there is a unique $a = \left(a_0, \dots, a_{\ell}\right) \in \mathcal{R}^{\ell+1}$ such that $v = \sum_{k=0}^{\ell} a_k v_k$. Due to the claim the only $(\bar{a}_0, \dots, \bar{a}_\ell) \in \R^{\ell+1}$ which fulfils $v = \sum_{k=0}^{\ell} \bar{a}_k v'_k$ is $\bar{a} = \varphi^{-1}(a)$, which implies that $V'$ is also a basis of $\mathcal{V}$.
\end{proof}

\end{inJournal}

\end{document}